\documentclass[12pt]{article}

\usepackage{amsmath, amssymb, amsthm, fullpage, charter, verbatim}
\newtheorem{theorem}{Theorem}
\newtheorem{lemma}{Lemma}
\newtheorem{corollary}{Corollary}

\newcommand{\eps}{\epsilon}

\newcounter{numquest}

\title{\textbf{Approximating Vertex Cover\\ in Dense Hypergraphs}\\[1ex]}

\author{
  Jean Cardinal\thanks{Universit\'e Libre de Bruxelles (ULB), CP212. Email:~\texttt{jcardin@ulb.ac.be}}\\[.75ex]
  Richard Schmied\thanks{Dept. of Computer Science, University of Bonn.
    Work supported by Hausdorff Doctoral Fellowship.
    Email:~\texttt{schmied@cs.uni-bonn.de}}
  \and
  Marek Karpinski\thanks{Dept. of Computer Science and the Hausdorff
    Center for Mathematics, University of Bonn.
    Supported in part by DFG grants and the Hausdorff Center grant EXC59-1.
    Email:~\texttt{marek@cs.uni-bonn.de}}\\[.75ex]
  Claus Viehmann\thanks{Dept. of Computer Science, University of Bonn.
    Work partially supported by Hausdorff Center for Mathematics, Bonn.
    Email:~\texttt{viehmann@cs.uni-bonn.de}}
}
\date{}

\begin{document}
\maketitle
\begin{abstract}
We consider the minimum vertex cover problem in hypergraphs in which every hyperedge has size $k$ (also known as {\em minimum hitting set} problem, or {\em minimum set cover} with element frequency $k$). Simple algorithms exist that provide $k$-approximations, and this is believed to be the best possible approximation achievable in polynomial time. We show how to exploit density and regularity properties of the input hypergraph to break this barrier. In particular, we provide a randomized polynomial-time algorithm with approximation factor $k/(1 + (k-1)\frac{\bar{d}}{k\Delta})$, where $\bar{d}$ and $\Delta$ are the average and maximum degree, respectively, and $\Delta$ must be $\Omega (n^{k-1} / \log n)$. The proposed algorithm generalizes the recursive sampling technique of Imamura and Iwama~(SODA'05) for vertex cover in dense graphs. As a corollary, we obtain an approximation factor $k/(2-1/k)$ for subdense regular hypergraphs, which is shown to be the best possible under the unique games conjecture.\\

\end{abstract}

\sloppy
\small

\section{Introduction}

A \emph{vertex cover} of a graph is a subset of its vertices hitting all edges. The \emph{Minimum Vertex Cover problem} asks to find a minimum size vertex cover, and is one of the most classical NP-complete problems. A \emph{$k$-uniform hypergraph} is a hypergraph whose edges are $k$-element subsets of his vertex set, and the \emph{$k$-uniform Vertex Cover problem} is the problem of finding a minimum size vertex cover in a $k$-uniform hypergraph. This problem is equivalent to the Set Cover problem where each element of the universe occurs in exactly $k$ sets. For $k=2$, it is the classical Vertex Cover problem in graphs. A simple $2$-approximation algorithm exists for this problem by constructing a maximal matching greedily. However, currently best known approximation algorithms can only achieve an approximation ratio of $2-o(1)$ \cite{H02,K09}. The problem is $k$-approximable in $k$-uniform hypergraphs, by choosing a maximal set of nonintersecting edges, and picking all vertices in them. The best known approximation factor is $k-(k-1)\ln \ln n /\ln n$ and is due to Halperin~\cite{H02}.\\

On the inapproximability side, one of the first hardness result for the $k$-uniform Vertex Cover problem is due to Trevisan~\cite{T01}. He obtained an inapproximability factor of $k^{1/19}$. Holmerin~\cite{H02a} proved that it is NP-hard to approximate within $k^{1-\eps}$, and in addition, in~\cite{H02b}, that the 4-uniform Vertex Cover problem is NP-hard to approximate within $2-\eps$. Dinur et al. proved a $(k-3-\eps)$ lower bound~\cite{DGK02}, later improved to $(k-1-\eps)$~\cite{DGKR05}.

In~\cite{K02}, Khot introduced the Unique Games Conjecture (UGC) as an approach to tackle the inapproximability of NP-hard optimization problems.
Assuming the unique games conjecture, an inapproximability factor of $k-\eps$ for the $k$-uniform Vertex Cover problem is due to Khot and Regev~\cite{KR08} which implies that the achieved ratios are the best possible. Recently, Bansal and Khot~\cite{BK10} were able to show in a conceptually different manner that this UGC-based inapproximability bound also holds on $k$-uniform hypergraphs that are almost $k$-partite.\\

In order to break these complexity barriers, we propose to consider dense instances of the Vertex Cover problem in hypergraphs, as was done previously for many other problems~\cite{AKK95,KRS09,KZ97,K01}. Typically, a graph is said to be dense whenever the number of edges is within a constant factor of $n^2$, where $n$ is the number of vertices. The Vertex Cover problem in dense hypergraphs, defined as hypergraphs with $\Omega ({n\choose k})$ hyperedges, has been considered before by Bar-Yehuda and Kehat~\cite{BK04}. They proposed an approximation algorithm which achieves a better approximation ratio than $k$ and showed that it is best possible under some assumptions similar to the UGC. To the authors' knowledge, this is the only result tackling the dense version of the Vertex Cover problem in hypergraphs. 

However, dense instances of the Vertex Cover problem in graphs ($k=2$) has been considered previously by Karpinski and Zelikovsky~\cite{KZ97}, Eremeev~\cite{E99}, Clementi and Trevisan~\cite{CT99}, and later by Imamura and Iwama~\cite{II05}. In the latter contribution, the authors give an approximation ratio parameterized by both the average and the maximum degree of the input graphs, that is strictly smaller than 2 whenever the ratio between the two is bounded, and the graph has average degree $\Omega (n / \log \log n)$. Other works tackling special instances of the Vertex Cover problem in hypergraphs include~\cite{K97,O05}.

\subsection{Definitions}

We use the notation $[i] := \{1,2,\ldots ,i\}$. Let $S$ be a finite set and $k\in [|S|]$, we introduce the abbreviation ${S \choose k}$ for the set of all subsets $S'\subseteq S$ such that $|S'|=k$. A {\em $k$-uniform hypergraph} is a pair $(V, E)$, where $V$ is the vertex set, and $E$ is a subset of ${V \choose k}$. We will usually set $n := |V|$ and $m := |E|$. In the remainder, unless stated explicitly, we will suppose that $k=O(1)$. 

A {\em Vertex Cover} of a $k$-uniform hypergraph $(V,E)$ is a set $C\subseteq V$ such that $e\cap C\not= \emptyset\ \forall e\in E$. The minimum Vertex Cover problem consists of finding a vertex cover of minimum size in a given hypergraph. 
 
The {\em degree} $d(v)$ of a vertex $v$ is equal to $|\{ e\in E : v\in e\} |$. Analogously, we define the degree $d(S)$ of a $S\subseteq V$ to be $|\{ e\in E : S\subseteq e\} |$. We refer to the average and maximum degree of a vertex as $\bar{d}$ and $\Delta$, respectively. A hypergraph is said to be {\em $d$-regular} whenever $d(v)=d\ \forall v\in V$, and {\em regular} whenever there exists a $d$ such that it is $d$-regular. 

We say that a $k$-uniform hypergraph is {\em $\ell$-wise $\eps$-dense} for $\ell + 1\in [k]$ and $\eps\in [0,1]$ whenever for every subset $S\in{V\choose \ell}$, we have $d(S)\geq \eps {n-\ell \choose k-\ell }$. Thus, for instance a 0-wise $\eps$-dense $k$-uniform hypergraph is a hypergraph with at least $\eps {n\choose k}$ hyperedges, while a 1-wise $\eps$-dense $k$-uniform hypergraph is such that every vertex is contained in at least $\eps {n-1 \choose k-1}$ hyperedges. This definition naturally generalizes the notion of {\em weak} and {\em strong density} in graphs defined in previous works~\cite{KZ97}. We call a $k$-uniform hypergraph {\em subdense} if 
$\bar{d} = \Omega \left( \frac{n^{k-1}}{\log n}\right)$ holds.

\subsection{Our Results}

We propose new approximation algorithms for the $k$-uniform Vertex Cover problem with approximation factors parameterized by the density and regularity parameters of the input hypergraph.\\ 

In 2004, Bar-Yehuda and Kehat~\cite{BK04} proved that the minimum vertex cover in 0-wise $\eps$-dense $k$-uniform hypergraphs was approximable within a factor 
$k/( k - (k-1)(1-\eps )^{\frac 1k} ) $. In the next section, we generalize this result to $\ell$-wise $\eps$-dense hypergraphs with $\ell>0$. We also provide a shorter proof of their result.\\

In Section~\ref{sec:rand}, we propose a randomized algorithm yielding an approximation factor $k/ (1 + (k-1)\frac{\bar{d}}{k\Delta })$ on $k$-uniform hypergraphs with $\bar{d} = \Omega (n^{k-1} / \log n)$. This implies a $k/(2-1/k)$-approximation algorithm for subdense regular hypergraphs. The proposed algorithm is based on Imamura and Iwama's recursive sampling technique~\cite{II05} for vertex cover in dense graphs. We generalize the technique to hypergraphs, and also improve its analysis.\\

Finally, in Section~\ref{sec:lb}, we prove the optimality of our bounds 
for a specified range of $\Delta$ under the Unique Games Conjecture~\cite{KR08}. Essentially, we show that any improvement on the approximation factor of our algorithms would yield an approximation factor asymptotically smaller than $k$ for the Vertex Cover problem in arbitrary $k$-uniform hypergraphs.

\section{Approximating Vertex Cover in Dense Hypergraphs}
\label{sec:approx}

We prove the following result:
\begin{theorem}
\label{thm:apxdense}
The Vertex Cover problem is approximable in polynomial time within a factor
$$
\frac k{k - (k-1)(1-\eps )^{\frac 1{k-\ell }}} -o(1)
$$
in $\ell$-wise $\eps$-dense $k$-uniform hypergraphs. 
\end{theorem}

In order to attain theorem~\ref{thm:apxdense}, we provide several lemmas which are crucial to the proof of the main result in this section. We start with the following simple extension of a Lemma from Bar-Yehuda and Kehat~\cite{BK04}. 

\begin{lemma}
\label{lem:subset}
Let $G = (V,E)$ be a $k$-uniform hypergraph with a minimum vertex cover $C$, and let $W\subseteq V$ such that
$|W\cap C|\geq \delta |W|$ for some constant $\delta\in [0,1]$. For every constant $j\in \mathbb{N}$, given $W$, 
we can compute an approximate solution to the Vertex Cover problem in $G$, with an approximation factor
\begin{equation*}
\frac {k}{1+(\delta k-1)\frac{|W|}{n-j}} ,
\end{equation*}
in polynomial time.
\end{lemma}
\begin{proof}
For a fixed $j\geq k$, we can check for all $i\in \{k-1,..,j\}$ and for all $S\in {V \choose n-i}$ whether $S$
is a vertex cover of $G$ in polynomial time. Thus, our solution can be assumed to be smaller than $n-j$.
Then, the algorithm computes a $k$-approximation in the hypergraph $G'$ induced by the edges that are not covered by $W$, and returns the union of this with $W$. Let us denote by $S$ the returned solution, by $C$ the considered optimal solution, by $\rho$ the approximation factor $|S|/|C|$, and by $C'$ the optimal solution of the Vertex Cover problem on $G'$. Clearly, we can assume $\frac{\rho |C|}{n-j}\leq 1$. Hence, we obtain:
\begin{eqnarray}
|S| =\rho |C| & \leq & |W| + k |C'| \\
    		& = & |W| + k (|C| - \delta |W|) \\
    		& = & k |C| - (\delta k-1)|W| \\
	 	& \leq & k |C| - (\delta k-1) \frac{|W|}{n-j} \rho |C| \\
\rho		& \leq & k - (\delta k-1) \frac{|W|}{n-j} \rho \\
\rho & \leq &  \frac k{1+(\delta k-1)\frac{|W|}{n-j}} .
\end{eqnarray}
\end{proof}

The following lemma also plays a key role in our analysis.

\begin{lemma}
\label{lem:degcond}
In a 0-wise $\eps$-dense $k$-uniform hypergraph $G$, the first $(1-(1-\eps )^{\frac 1k})n$ highest-degree vertices all have degree at least $(1-(1-\eps )^{\frac {k-1}k}) {n-1\choose k-1}$.
\end{lemma}
\begin{proof}
We consider a hypergraph $G$ containing $m\geq \eps {n\choose k}$ hyperedges.
Let us denote by $H$ the set of the first $(1-(1-\eps )^{\frac 1k})n$ highest-degree vertices (breaking ties arbitrarily).
Suppose the statement is not true. Then the number $m$ of edges in $G$ is strictly smaller than the number of edges in a hypergraph in which all vertices of $H$ have degree ${n-1\choose k-1}$ (the maximum possible), and all the remaining edges have degree exactly $(1-(1-\eps )^{\frac {k-1}k}) {n-1\choose k-1}$. Thus
\begin{eqnarray}
m & < & \frac 1k \left( |H| {n-1\choose k-1} + (n - |H|) (1-(1-\eps )^{\frac {k-1}k}) {n-1\choose k-1} \right) \\
    & = & \frac 1k \left( (1-(1-\eps )^{\frac 1k})n {n-1\choose k-1} + (n - (1-(1-\eps )^{\frac 1k})n) (1-(1-\eps )^{\frac {k-1}k}) {n-1\choose k-1} \right) \\
    & = & (1-(1-\eps )^{\frac 1k}) {n\choose k} + (1-\eps )^{\frac 1k} (1-(1-\eps )^{\frac {k-1}k}) {n\choose k}  \\
    & = & \eps {n\choose k} 
\end{eqnarray}
which is a contradiction, since $G$ is weakly $\eps$-dense.
\end{proof}

We now prove Theorem~\ref{thm:apxdense}, and first consider the case $\ell = 0$. In order to apply Lemma~\ref{lem:subset}, we need to find a large subset $W$ of a minimum vertex cover. The following recursive algorithm returns a polynomial-size collection $\cal W$ of subsets $W_i\subseteq V$, such that at least one of them is contained in a minimum vertex cover.\\

input: a 0-wise $\eps$-dense $k$-uniform hypergraph $G = (V,E)$
\begin{enumerate}
\item if $k=1$ then 
\begin{enumerate}
\item return a minimum vertex cover of $G$, of size $|E|\geq \eps n$
\end{enumerate}
\item else:
\begin{enumerate}
\item let $H$ be the set of the first $(1-(1-\eps )^{\frac 1k})n$ highest-degree vertices (breaking ties arbitrarily)
\item add $H$ to $\cal W$
\item for each $v\in H$:
\begin{enumerate}
\item let $G'$ be the $(k-1)$-uniform hypergraph $(V-\{ v \}, \{ e-\{v\} : e\in E, v\in e \} )$
\item let $\eps ' := 1 - (1-\eps )^{\frac k{k+1}}$
\item call the procedure recursively, with the parameters $G', \eps', k-1$; let ${\cal W}'$ be its output
\item add the sets of ${\cal W}'$ to $\cal W$
\end{enumerate}
\item return $\cal W$
\end{enumerate}
\end{enumerate}

\begin{lemma}
\label{lem:extract}
Given a 0-wise $\eps$-dense $k$-uniform hypergraph $G$, we can find in polynomial time a set ${\cal W} := \{ W_i\}_{i=1}^s$ of size $s = O(n^k)$, with $W_i\subseteq V$, and such that
\begin{enumerate}
\item There exists $i\in [s]$ such that $W_i$ is a subset of a minimum vertex cover of $G$,
\item $|W_i|\geq (1-(1-\eps )^{\frac 1k})(n-k+1),\ \forall i\in [s]$.
\end{enumerate} 
\end{lemma}
\begin{proof}
The above algorithm outputs of set $\cal W$ of size $O(n^k)$ in time $O(n^k)$, which is polynomial since we assumed $k=O(1)$. 

The first condition is verified by induction. If all vertices in $H$ belong to a minimum vertex cover, we are done. Otherwise, there exists $v\in H$ that does not belong to any minimum vertex cover. But then a minimum vertex cover of $G$ must contain a minimum vertex cover of the $(k-1)$-uniform hypergraph $G'$, as otherwise some edges will not be covered. By induction, the recursive call returns one subset contained in a minimum vertex cover of $G'$, hence also in a minimum vertex cover of $G$. The base case $k=1$ is trivial.

We prove the second property by induction as well. Suppose that $|W_i|\geq (1-(1-\eps )^{\frac 1k})(n-k+1),\ \forall i\in [s]$ holds for $k$-uniform hypergraphs, for some fixed value of $k$. We now prove the property for $k+1$. From Lemma~\ref{lem:degcond}, the recursive calls are performed on weakly $\eps'$-dense hypergraphs with $n-1$ vertices. Thus by the induction hypothesis, the recursive call returns a collection of sets $W_i$ of size
\begin{eqnarray}
|W_i| & \geq &  (1-(1-\eps' )^{\frac 1k})((n-1)-k+1) \\
	  & = & (1-(1- (1 - (1-\eps )^{\frac k{k+1}}) )^{\frac 1k})(n -(k+1) -1) \\
	  & = & (1 - (1-\eps)^{\frac 1{k+1}} ) (n -(k+1) - 1),
\end{eqnarray} 
as claimed. The base case $k=1$ is verified, as in that case the procedure yields at least $\eps n$ vertices.
\end{proof}

We now tackle the case $\ell >0$, i.e. hypergraphs in which every subset of $\ell$ vertices is contained in $\eps {n-\ell \choose k-\ell }$ hyperedges.

\begin{lemma}
\label{lem:extractstrong}
Given a $\ell$-wise $\eps$-dense $k$-uniform hypergraph $G$, we can find in polynomial time a set ${\cal W} := \{ W_i\}_{i=1}^s$ of size $s=O(n^k)$, with $W_i\subseteq V$, and such that
\begin{enumerate}
\item There exists $i\in [s]$ such that $W_i$ is a subset of a minimum vertex cover of $G$,
\item $|W_i|\geq (1-(1-\eps )^{\frac 1{k-\ell }})(n-k+1),\ \forall i\in [s]$.
\end{enumerate} 
\end{lemma}
\begin{proof}
Let $V$ be the vertex set of $G$. Consider a subset $S$ of $\ell$ vertices that do not belong to a given minimum vertex cover $C$ of $G$. We define the hypergraph $G'$ as the subhypergraph of $G$ whose vertex set is $V':=V\setminus S$, and whose hyperedges are the hyperedges of $G$ containing $S$, restricted to $V\setminus S$. Since $S$ is not contained in $C$, $C$ must contain a vertex cover of $G'$. From the definition of $\ell$-wise $\eps$-density, $G'$ has at least $\eps {n-\ell\choose k-\ell}$ edges. Hence $G'$ is a 0-wise $\eps$-dense $(k-\ell$)-uniform hypergraph with $n-\ell$ vertices. From Lemma~\ref{lem:extract}, we can extract $O(n^{k-\ell})$ candidates $W_i$, which are subsets of $V'$ of size at least $(1 - (1-\eps)^{\frac 1{k-\ell}} ) (n-k+1)$. One of them at least is contained in a minimum vertex cover of $G'$, and hence of $G$. By enumerating all $O(n^{\ell})$ possibilities for $S$, we get the result in time $O(n^k)$.
\end{proof}

The proof of Theorem~\ref{thm:apxdense} is now straightforward. By testing all possible sets $W_i\in{\cal W}$ and choosing the one that yields the smallest cover,  
we get from Lemma~\ref{lem:subset} (with $\delta = 1$) a polynomial-time approximation algorithm with approximation factor
\begin{equation*}
\frac {k}{1+(k-1)\frac{|W|}{n-j}} = \frac k{1+(k-1) (1 - (1-\eps)^{\frac 1{k - \ell}})  }-o(1) = \frac k{k - (k-1)(1-\eps )^{\frac 1{k - \ell}}} -o(1)
\end{equation*}
in $\ell$-wise $\eps$-dense $k$-uniform hypergraphs. 

\section{A Randomized Approximation Algorithm for Near-Regular Hypergraphs}
\label{sec:rand}

We now prove a complementary result, yielding an approximation factor for the problem that is parameterized by both the average and the maximum degree of the input hypergraph. Our algorithm and analysis extend the work of Imamura and Iwama~\cite{II05}. 

\begin{theorem}
\label{thm:mainsampling}
For every $\eps>0$ and $k=o(\log n)$, there is a randomized approximation algorithm which computes with
high probability a solution for the Vertex Cover problem on $k$-uniform hypergraphs with approximation ratio
$$
\frac{k}{1+(k-1)\left. \frac{\bar{d}}{k\Delta}\right.}+\eps,
$$ 
where $\bar{d}$ and $\Delta$ denote the average and maximum degree of the hypergraph, respectively.
The running time is $n^{O(1)}2^{O(k\cdot\psi(n)\log\log k)}$, which is polynomial if $k\cdot\psi(n)\log\log k = O(\log n)$ and
quasi-polynomial if $\psi(n)=\mathrm{polylog\ } n$.
\end{theorem}

By letting $\bar{d}/\Delta \to 1$ and $\bar{d} = \Omega \left( \frac{n^{k-1}}{\log n}\right)$, we get the following corollary.
\begin{corollary}
The Vertex Cover problem is approximable in polynomial time within a factor $\frac {k}{2-1/k}$ in subdense regular $k$-uniform hypergraphs with $k=O(1)$.
\end{corollary}

\subsection{Outline}

The algorithm first iteratively removes vertex subsets, until a sufficiently small set of vertices remain, then applies the trivial $k$-approximation algorithm on the remaining induced hypergraph. We first suppose that at every step $i$ of the algorithm, we are able to guess a sufficiently large subset of an optimal solution of the current hypergraph $G_i$. This subset of vertices is removed, together with the edges that they cover, to form $G_{i+1}$. We will see in the next subsection how we can sample the set $\cal W$ computed in the algorithm of section~\ref{sec:approx} to perform this guessing step efficiently. The union of the removed sets will form the set $W$ allowing us to use Lemma~\ref{lem:subset}. We aim at obtaining such a set $W$ of size approximately $\beta n$, with $\beta := \bar{d} /(k\Delta )$.

Letting $G_i$ be the hypergraph considered at the $i$th step, we denote by $n_i$ its number of vertices, by $E_i$ its edge set, by $\eps_i := |E_i| / {n_i\choose k}$ its density (in the $\ell = 0$ sense), and by $\bar{d}_i := |E_i|k/n_i$ its average degree. We also let $\psi (n) := {n\choose k-1} / \Delta$. Let $s_i := n_i - (1-\beta )n$. Note that $s_i = 0\Rightarrow n_i=(1-\beta )n \Rightarrow n-n_i = \beta n$. Since $n-n_i$ is the size of the extracted set $W$, $s_i$ can serve as a measure of progress of the procedure. At every step, we remove $c  n /\psi (n)$ vertices, until $s_i\leq c  n$, for a small constant $c \in [0,1]$. Thus at the end of the procedure, we will have $s_i\leq c n \Rightarrow |W|\geq (\beta - c )n$. We now show we can always find a set of this size contained in a minimum vertex cover. From Lemma~\ref{lem:extract}, we know there exists such a set of size at least $r_i$, where
$$
r_i := \left( 1-\left( 1-\eps_i \right)^{\frac 1k}\right) (n_i-k+1) .
$$

We suppose that we can efficiently guess this subset, and prove that it is large enough.

\begin{lemma}
\label{lem:extract2}
provided $s_i \geq c  n$, we have the following inequality:
\begin{equation}
\label{eq:numvtx}
r_i \geq c \frac{n}{\psi (n)} .
\end{equation}
\end{lemma}

We need three simple intermediate results.
\begin{lemma}
\label{lem:maxedges}
Let $H$ be a $k$-uniform hypergraph with $n$ vertices, maximum degree $\Delta$ and maximum independent set of size at least $x$. Then 
$H$ has at most $\Delta (n-x)$ edges.
\end{lemma}
\begin{proof}
Every vertex covers at most $\Delta$ edges, hence the size of a minimum vertex cover of $H$, say $\tau$, satisfies $|E(H)| \leq \Delta \tau$. Also, by definition, the largest independent set of $H$ has size $n-\tau \geq x$, hence $\tau \leq n - x$.
\end{proof}

\begin{lemma}
\label{lem:ubei}
$|E_i| \geq \Delta s_i$
\end{lemma}
\begin{proof}
The hypergraph $G' := G - E_i$ has an independent set of size at least $n_i = s_i + (1-\beta )n$, as by definition all the vertices of $G_i$ form an independent set in $G'$. Thus from Lemma~\ref{lem:maxedges}, $G'$ can have at most $\Delta (n-(s_i + (1-\beta )n)) = \Delta (\beta n - s_i)$ edges. So we obtain 
$$
|E_i| = |E| - |E(G')| \geq \frac{\bar{d}n}{k} - \Delta (\beta n-s_i) = \Delta \beta n - \Delta (\beta n -s_i) = \Delta s_i .
$$
\end{proof}

\begin{lemma}
\label{lem:boundr}
$1-\left( 1-\eps \right)^{\frac 1k} \geq \frac{\eps}{k}\ \forall \eps\in [0,1], k \geq 1$.
\end{lemma}

\begin{proof}[Proof of Lemma~\ref{lem:extract2}]
Combining the two previous lemmas, we obtain:
\begin{eqnarray}
\frac{r_i}{s_i} & \geq & \frac{\eps_i (n_i-k+1)}{k s_i}\ \ \text{(from\ Lemma~\ref{lem:boundr})} \\
	& = & \frac{|E_i| (n_i-k+1)}{k s_i {n_i \choose k}} \\
	& \geq & \frac{\Delta s_i (n_i-k+1)}{k s_i {n_i \choose k}}\ \ \text{(from\ Lemma~\ref{lem:ubei})} \\
	& = & \frac{\Delta}{{n_i \choose k - 1}} \geq \frac{\Delta}{{n \choose k - 1}} = \frac 1{\psi (n)} \\
r_i & \geq & \frac{s_i}{\psi(n)} \geq c \frac n{\psi (n)}.	 
\end{eqnarray}
\end{proof}

Thus we know we can extract this number of vertices at each step. The number $t$ of required steps is therefore
\begin{equation}
t := \frac{(\beta - c ) n}{c \frac n{\psi (n)}} = \psi (n) (\beta / c  - 1) .
\end{equation}

\subsection{The recursive sampling procedure}

We now define a recursive sampling procedure that will allow us to efficiently guess the subsets. The procedure {\bf $IR$} returns a small set of candidate subsets. It is a sampling version of the procedure given in the previous section.\\

\noindent\textbf{Procedure $IR(G,l)$} (Inner Recursion)\\[1ex]
Input: a $k$-uniform hypergraph $G = (V, E)$ and $l\in \mathbb{N}$
\begin{enumerate}
\item ${\cal W} \gets \emptyset$
\item  if $k = 1$ then
\begin{enumerate}
\item return $\{ C \}$, where $C$ is set of $c n/\psi(n)$ arbitrary vertices of $E$ 
\end{enumerate}
\item else:
\begin{enumerate}
\item let $H$ be the set of the first $c n/\psi(n)$ highest-degree vertices (breaking ties arbitrarily)
\item add $H$ to $\cal W$
\item let $H'\subseteq H$ be a random subset of $l$ vertices
\item for each $v\in H'$: 
\begin{enumerate}
\item let $G'$ be the $(k-1)$-uniform hypergraph $(V\setminus \{v\}, \{e\setminus\{v\}: e\in E, v\in e\})$
\item call the procedure recursively with the parameters $k-1, G', l$; let ${\cal W}'$ be its output
\item add the sets of ${\cal W}'$ to ${\cal W}$
\end{enumerate}
\item return $\cal W$
\end{enumerate}
\end{enumerate}

The procedure {\bf $ER$} below iterates this extraction until $s_i\leq c  n$, using $t$ recursion levels. It is initially called with $i = 0$ and $t=\psi (n) (\beta / c  - 1)$. It also uses a variable $l$, which sets the size of the sample used.\\

\noindent\textbf{Algorithm $ER(G, i)$} (Outer Recursion)\\[1ex]
Input: a $k$-uniform hypergraph $G$ with $i\in \mathbb{N}$, $i\leq t$
\begin{enumerate}
\item ${\cal W} \gets \emptyset$
\item if $i<t$ then:
\begin{enumerate}
\item ${\cal W} \gets IR(G, l)$
\item return $\min\{ W' \cup ER(G \setminus W',i+1)\mid W'\in {\cal W} \}$
\end{enumerate}
\item else ($i=t$)
\begin{enumerate}
\item apply a $k$-approximation algorithm to $G$ and let $C$ be the resulting vertex cover
\item return $C$
\end{enumerate}
\end{enumerate}

Let us choose a constant $p\in (0,1)$, and define the sample size $l$ as 
$$
l := \lceil \log (1-p^{\frac 1k}) / \log p\rceil .
$$ 
Note that $1-p^{\frac 1k} = \Theta (\frac 1k)$, hence $l = \Theta (\log k)$.
With this value of $l$, the procedure $IR$ has the following property. 
We denote by $C$ an arbitrary vertex cover of the input hypergraph.
\begin{lemma}
The set $\cal W$ returned by the procedure $IR$ contains a subset $W'$ such that 
$|W'\cap C| \geq p |W'|$ with probability at least $p$.
\end{lemma}
\begin{proof}
Let $H$ be the first $c n/\psi(n)$ highest-degree vertices. 
If $|H\cap C| \geq p |H|$ then we are done. Thus we can suppose that $|H\cap C| < p |H|$, and 
the probability that a random vertex of $H$ belongs to $C$ is at most $p$. Thus with probability
at least $1-p^l$ we get a vertex $v\not\in C$ in the selected sample, and $C$ must contain a vertex cover
of the hypergraph $G'$ defined by $v$. By iterating, we eventually get that the probability is at least 
$(1-p^l)^k \geq p$, from the definition of $l$.
\end{proof}

The procedure $ER$ performs a recursive exploration of a search tree, branching on every subset $W'$ in the set of candidates $\cal W$.
A root-to-leaf path in this tree yields a set $W$, defined as the union of all the candidates $W'$ selected along the path.  
We now prove that with high probability, this search tree contains a path yielding a suitable set $W$.

\begin{lemma}
\label{lem:Wprop}
For any $\delta >0$, the procedure $ER$ constructs a set $W$ of $(\beta - c) n$ vertices, 
such that $| W \cap C | \geq (1-\delta) p^2 |W|$, with probability at least
$1 - e^{-\psi(n)(\beta/c - 1) p \frac{\delta^2}2}$.
\end{lemma}
\begin{proof}

We introduce the random variable $X_i$ denoting the success in the $i$th step. Therefore, we
have $p(X_i=1) = p$ and $p(X_i=0) = 1-p$. We let $X = \sum_{i\in [t]}X_i $.
Clearly, we have
$$
E[X]\geq  t p = \psi(n)(\beta/c -1) p .
$$
Since $cn/\psi (n)$ vertices are chosen at every step, the expected number of vertices of $W$ that are contained in $C$ is 
\begin{eqnarray}
E[X] p c \frac n{\psi (n)} & \geq & (\beta/c - 1) p^2 c n \\
			 & =    & (\beta - c) p^2 n .
\end{eqnarray}
The claimed statement is obtained using Chernoff bounds.
\end{proof}

\subsection{Assembling the pieces}

The proof of Theorem~\ref{thm:mainsampling} follows directly from the previous lemmas. 

\begin{proof}[Proof of Theorem~\ref{thm:mainsampling}]
The algorithm is as follows. First, select a constant $c < 1$, arbitrarily small, 
and a probability $p$, that can be arbitrarily close to 1. Then compute the corresponding 
value of the sample size $l$ and the number of steps $t$, and run the procedure $ER$ with these parameters. 
From Lemmas~\ref{lem:Wprop} and \ref{lem:subset}, with probability at least 
$1-e^{-\psi(n)(\beta/c -1) p \frac{\delta^2}2}$, the procedure $ER$ achieves an approximation ratio
\begin{equation*}
\frac {k}{1 + ((1-\delta)p^2 k - 1) (\beta - c)}.
\end{equation*}
When $c\to 0$ and $p\to 1$, this ratio is arbitrarily close to $k/(1+(k-1)\beta)$.

The procedure $ER$ generates a search tree of height $t$ and fan-out less than $l^k$. At every node of the tree,
the procedure $IR$ is called, taking $O(n^{O(1)})+O(l^k)$ time. The overall running time of the procedure is therefore
$O(n^{O(1)} \cdot l^{kt})$. We have $kt = k \psi (n) (\beta / c  - 1) = \Theta(k\psi(n))$ and $l = \Theta(\log k)$. 
Hence the running time is $n^{O(1)}2^{O(k\psi(n)\log\log k)}$, as claimed.
\end{proof}

\section{Lower bounds}
\label{sec:lb}

In this section, we provide several inapproximability results based on two well-known conjectures, namely the UGC and $P\neq NP$. In particular, we show that the achieved approximation ratios in section~\ref{sec:approx} and \ref{sec:rand} are optimal in a specified range of $\Delta$ assuming the Unique Game Conjecture.\\

We introduce two hardness results on which our inapproximability bounds are based. The following hardness result is due to Khot and Regev~\cite{KR08} and is underlying the UGC.

\begin{theorem}\label{ugcmain}
Given a $k$-uniform hypergraph $H=(V,E)$, let $OPT$ denote an optimal vertex cover of $H$.
For every $\delta>0$, the following is UG-hard to decide:
$$|V|(1-\delta)\leq |OPT| \textrm{  or  } |OPT|\leq |V|(1/k+\delta)$$
\end{theorem}

On the other hand, Dinur et al.~\cite{DGKR05} proved the following NP-hardness result.

\begin{theorem}\label{npmain}
Given a $k$-uniform hypergraph $H=(V,E)$ with $k\geq 3$, let $OPT$ denote an optimal vertex cover of $H$.
For every $\delta>0$, the following is NP-hard to decide:
$$|V|(1-\delta)\leq |OPT| \textrm{  or  } |OPT|\leq |V|(1/(k-1)+\delta)$$
\end{theorem}

The next theorem shows that the approximation algorithm from theorem~\ref{thm:apxdense} is optimal assuming the UGC.

\begin{theorem}
Assuming the UGC and $P\neq NP$, respectively, there is no polynomial time algorithm with an approximation ratio better than
$$
 \frac{k}{k+(k-1)(1-\eps)^{\frac 1{k-\ell}}}\quad \textrm{ and for any }k\geq 3\quad  
\frac{k-1}{k-1+(k-2)(1-\eps)^{\frac 1{k-\ell}}} \textrm{  ,   } 
$$ 
respectively, by a constant for the Vertex Cover problem in $\ell$-wise $\eps$-dense $k$-uniform hypergraphs. 
\end{theorem}
\begin{proof}
As a starting point of the reduction, we use the hypergraph $H=(V,E)$ from theorem~\ref{ugcmain}
and construct a $\ell$-wise $\eps$-dense $k$-uniform hypergraph $H'=(V',E')$. We introduce 
the abbreviations $N:=|V'|$ and $n:=|V|$.
First, we join a clique $C=(V(C),{V(C) \choose k})$ of size $c\cdot N$ (that is, $c\cdot  N$ vertices with all possible hyperedges of size $k$)
 to $H$. Furthermore, we add all hyperedges $e\in {C\cup V \choose k}$, such that at least one vertex in $e$ 
 is from $C$. Thus, we obtain $N=c N+n$ and $N=n/(1-c)$. We denote by $OPT'$
an optimal vertex cover of $H'$. The UG-hard decision question from theorem~\ref{ugcmain}
transforms into the following:
$$
n(1-\delta)+\frac{c n}{1-c}\leq |OPT'| \textrm{  or  } 
|OPT'|\leq n(\frac{1}{k}+\delta)+\frac{c n}{1-c} 
$$

Assuming the UGC, this implies the hardness of approximating the Vertex Cover problem
in $\ell$ wise $\eps$-dense hypergraphs for every $\delta'>0$ to within:
\begin{eqnarray}
\frac{n(1-\delta)+\frac{c n}{1-c}}{n(1/k+\delta)+\frac{c n}{1-c}}
 & = & \frac{n(1-\delta)(1-c)+nc}{n(1/k+\delta)(1-c)+c n}
   = \frac{1-\delta(1-c)}{1/k-c/k+c+\delta/k(1-c)} \\
 & = & \frac{k}{1+(k-1)c}-\delta' \label{eq:UGCc}
\end{eqnarray}   
By using Theorem~\ref{npmain}, an analog calculation leads to the fact that it is NP-hard to approximate 
the Vertex Cover problem
in $\ell$ wise $\eps$-dense hypergraphs for every $\delta'>0$ to within:
\begin{eqnarray}
&&\frac{k-1}{1+(k-2)c}-\delta' \label{eq:NPc}
\end{eqnarray} 

In order to obtain a $\ell$-wise $\eps$-dense $k$-uniform hypergraph, we have to determine the right 
value for $c$.
Let $S\in {V \choose \ell}$ be a subset of $V$, and $d(S)$ denote its degree in $H'$. 
We can assume that $\min \{d(S')\mid S'\in { V' \choose \ell}\} =\min \{d(S)\mid S\in { V \choose \ell}\}$.
Hence, it suffices to ensure that $ d(S)\geq \epsilon {N-\ell \choose k-\ell}=\epsilon \frac{(N-\ell)^{k-\ell}}{(k-\ell)!}-o\left(\frac{(N-\ell)^{k-\ell}}{(k-\ell)!}\right) $ holds. We obtain
\begin{eqnarray*}
d(S)	& \geq &  { N- \ell \choose k-\ell } - { N-\ell -c N \choose k-\ell } \\
	& \geq & \frac{(N- k+1)^{k-\ell }}{(k-\ell )!}-\frac{(N-\ell - c N)^{k-\ell}}{(k-\ell )!} \\
	& =    &\frac{N^{k-\ell}[(1- (k+1)/N)^{k-\ell}-(1- \ell / N - c)^{k-\ell}]}{(k-\ell )!}
\end{eqnarray*}

By comparing the factors, we get $(1-(k+1)/N)^{k-\ell}-(1-\ell/n'-c)^{k-\ell} = \eps - o(1)$.
 Therefore, we deduce that 
$H'$ is $\ell$-wise $\eps$-dense for 
$$
c = 1-[(1- k/n')^{k-\ell} - \eps ]^{\frac 1{k-\ell}} + o(1) = 1 - (1-\eps )^{\frac 1{k-\ell}} + o(1).
$$
Plugging $c$ into equation~(\ref{eq:UGCc}), we obtain the following inapproximability factor $R$
assuming the UGC.
\begin{eqnarray*}
R & = & \frac{k}{1+(k-1)c}-\delta'\\
  & =    & \frac{k}{k+(k-1)(1-\eps)^{\frac 1{k-\ell}}}-o(1)-\delta'
\end{eqnarray*}
 Analogously, by using equation~(\ref{eq:NPc}) instead, we get for every $\delta'>0$ the inapproximability factor  
 $$\frac{k-1}{1+(k-2)(1-\eps)^{\frac 1{k-\ell}}}-o(1)-\delta'$$ assuming $P\neq NP$.

\end{proof}

The next theorem generalizes the former result in the sense that the inapproximability factor is
parameterized by the maximum degree $\Delta$ of the hypergraph.
In addition, it shows that the approximation algorithm from theorem~\ref{thm:mainsampling} is optimal in a specified range of
$\Delta$ assuming the UGC.

\begin{theorem}\label{theo:ave deg}
Assuming the UGC and $P\neq NP$, respectively, for every $c \in \mathbb{N}$, there is a $c_2\in [0,1]$ such that
 no polynomial time algorithm can find a solution with an approximation ratio better than
$$
\frac{k}{1+(k-1)\left. \frac{\bar{d}}{k\Delta}\right.} \textrm{   and for any }k\geq 3 \quad
 \frac{k-1}{1+(k-2)\left. \frac{\bar{d}}{k\Delta}\right.},
$$ 
respectively, by a constant for the Vertex Cover problem $k$-uniform hypergraphs with 
average degree $\bar{d}$, maximum degree 
$\Delta=\Omega(n^{\frac{k-1}{c }})$ and $\Delta\leq c_2\frac{n^{k-1}}{k^{k-1}(k-1)!}$. 
\end{theorem}

\begin{proof}
We proceed as before, by considering the hypergraph $H=(V,E)$ of Theorem~\ref{ugcmain}, with $n:=|V|$. 
Since the only way to achieve a better factor than $k$ is in case of $\frac{\bar{d}}{\Delta}\in (0,1]$, 
we will set $\frac{\bar{d}}{\Delta}:=\epsilon \in (0,1]$ in the remainder.
We construct a new hypergraph $H'=(V',E')$ consisting of $(1-\epsilon/k)n$ disjoint copies of $H$, together with $n\epsilon/k$ disjoint complete $k$-uniform hypergraphs of size $n+k$ (cliques). We let $V_1$ be the set of vertices in the copies of $H$, and $V_2$ the set of vertices of the cliques, and $n':=|V'|=|V_1|+|V_2|=n^{2}+o(n^2)$. Note that the degrees of the vertices restricted to neighbors in $V_1$ or $V_2$ are at most 
${n+k-1\choose k-1}=O(n^{k-1})=O(\sqrt{n'^{k-1}})$. Hence, we can make $H'$ to have asymptotically the average degree $\bar{d}=\omega(n'^{\frac{k-1}{2}})$ and maximum degree $\Delta=\omega(n'^{\frac{k-1}{2}})$ by adding as many hyperedges as needed with vertices in both $V_1$ and $V_2$, thus defining $E'$. By construction, the maximal degree which can be obtained is proportional to ${n^2\epsilon/k \choose k-1}$.

By definition, a vertex cover of $H'$ must contain at least $n+k-k=n$ vertices of each clique. Thus, we need to include at least $n^{2}\epsilon / k $ vertices. We now consider the two cases in the decision problem above. If a vertex cover of $H$ requires $n(1-\delta)$ vertices, then we need $(1-\epsilon/k)n\cdot n(1-\delta ) =
 (1-\epsilon/k)n^2(1-\delta)$ additional vertices to cover all the copies. In the other case, $(1-\epsilon/k)n^2(1/(k-1)+\delta) $ vertices suffice. Up to a $O(n)$ term, those vertices are sufficient to cover $H'$. Denoting by $OPT'$ an optimal vertex cover of $H'$, the UG-hard decision question from theorem~\ref{npmain} therefore becomes:
$$
\left(1-\epsilon / k\right) n^2(1-\delta) + \frac{n^2\epsilon}k + o(n^2) \leq |OPT'| \textrm{  or  } 
\left(1-\epsilon /k\right) n^2(1/(k-1)+\delta) + \frac{n^2\epsilon}k + o(n^2) \geq |OPT'|
$$
Under the UGC, this implies the hardness of approximating within a factor:
\begin{eqnarray}
\frac{\left(1-\epsilon /k\right) n^2(1-\delta) + \frac{n^2\epsilon}k }{\left(1-\epsilon/k\right)n^2(1/(k-1)+\delta) 
+ \frac{n^2\epsilon}k }
& = & \frac{(k-1)\left(1-\epsilon /k\right)(1-\delta) + (k-1)\frac{\epsilon}k }{(k-1)\left(1-\epsilon/k\right)(1/(k-1)+\delta) 
+ (k-1)\frac{\epsilon}k } \\
& = & \frac{k-1}{1-\epsilon/k +(k-1)\epsilon/k}  - \delta '\\
& = & \frac{k-1}{1+(k-2)\epsilon/k}  - \delta '.
\end{eqnarray}
Since we can add $O(n^c)$ copies of $H$ with an arbitrary $c=O(1)$ to construct $H'$, the result follows.
\end{proof}

A similar construction with $\epsilon\to 1$ yields the following.

\begin{theorem}
Assuming the UGC and $P\neq NP$, respectively there is no polynomial time algorithm yielding an approximation ratio better than $\frac{k}{2-1/k}$ and 
for $k\geq 3$ than $\frac{k}{2}$, respectively, by a constant for the Vertex Cover problem in $k$-uniform regular hypergraphs.
\end{theorem}

\end{document}